\documentclass[11pt,twoside]{article}

\usepackage[ansinew]{inputenc} 
\usepackage{amsmath,amsfonts,amssymb,amsthm}
\usepackage{pstricks,pst-node,pst-coil,pst-plot,pstricks-add}
\usepackage{graphics,geometry,epsfig}
\usepackage{bbm}

\newtheorem{theorem}{Theorem}[section]

\newtheorem{lemma}[theorem]{Lemma}

\newtheorem{proposition}[theorem]{Proposition}

\newcommand{\DETAILS}[1]{}

\newcommand{\supp}{\operatorname{supp}}

\newcommand{\Curl}{\operatorname{curl}}

\newcommand{\R}{\mathbb{R}}

\newcommand{\N}{\mathbb{N}}

\newcommand{\cH}{{\cal{H}}}

\newcommand{\cE}{\mathcal{E}}

\newcommand{\cQ}{\mathcal{Q}}

\newcommand{\eps}{{\varepsilon}}

\newcommand{\vphi}{{\varphi}}
\newcommand{\e}{{\varepsilon}}
\newcommand{\vep}{{\varepsilon}}           

\newcommand{\po}{{x}}
\newcommand{\posp}{{\underline{x}}}

\newcommand{\nc}{\newcommand}
\nc{\G}{\Gamma} \nc{\g}{\gamma} \nc{\Omt}{\tilde{\Omega}}
\nc{\ta}{\tau} \nc{\w}{\omega} \nc{\io}{\iota} \nc{\h}{\theta}
\nc{\Si}{\Sigma}

\nc{\be}{\begin{equation}} \nc{\la}{\label} \nc{\ba}{\begin{array}}
\nc{\ea}{\end{array}} \nc{\bs}{\begin{split}} \nc{\es}{\end{split}}

\nc{\p}{\partial} \nc{\ra}{\rightarrow} \nc{\ran}{\rangle}
\nc{\lan}{\langle}

\nc{\bP}{\bar{P}} \nc{\bQ}{\bar{Q}} \nc{\bL}{\bar{L}} \nc{\1}{{\bf
1}}

\begin{document}
\title{\ Multipolarons in a Constant Magnetic Field}
\author{ I.~Anapolitanos\footnote{Supported by DFG under Grant GR 3213/1-1.}~~and M.~Griesemer \\ Universit\"at Stuttgart, Fachbereich Mathematik \\ 70550 Stuttgart, Germany}
\date{}
\maketitle

\vspace{-5mm}
\begin{center}
\emph{Dedicated to the memory of Walter Hunziker}
\end{center}

\abstract{The binding of a system of $N$ polarons subject to a
constant magnetic field of strength $B$ is investigated within the
Pekar-Tomasevich approximation. In this approximation the energy of
$N$ polarons is described in terms of a non-quadratic functional
with a quartic term that accounts for the electron-electron
self-interaction mediated by phonons. The size of a coupling
constant, denoted by $\alpha$, in front of the quartic is determined
by the electronic properties of the crystal under consideration, but
in any case it is constrained by $0<\alpha<1$. For all values of $N$
and $B$ we find an interval $\alpha_{N,B}<\alpha<1$ where the $N$
polarons bind in a single cluster described by a minimizer of the
Pekar-Tomasevich functional. This minimizer is exponentially
localized in the $N$-particle configuration space $\R^{3N}$.}

\section{Introduction}
The electron-phonon interaction in a polar crystal mediates an
interaction between pairs of electrons which becomes an
electrostatic Coulomb attraction in the Pekar-Tomasevich
approximation. This attraction competes with the Coulomb repulsion
between the equally charged electrons, and the question arises
whether $N$ electrons may form a bound cluster. Due to the
constraint on the parameters of the model, the $1/|x|$-part of the
electron-electron interaction is repulsive. There remains, however
an attractive short range interaction, which seems to be of van der
Waals type and which may lead to $N-$particle bound states
\cite{Le}. This phenomenon of bound \emph{multipolarons} had
previously been observed in Fr\"ohlich's large polaron model on
which the Pekar-Tomasevich approximation is based
\cite{VPD1990,BKD2005}. Similarly, the binding of polarons subject
to a constant magnetic field had been investigated within the
Fr\"ohlich model \cite{BD1996}. Yet, in that case, the analysis in
the literature is based on poorly justified variational estimates,
and the conclusions remain doubtful. The present paper establishes,
within the Pekar-Tomasevich approximation, the existence of bound
$N$-polaron clusters in a constant magnetic field of any strength.
It is a continuation of a previous work of one of us, concerning the
case $N=2$ \cite{GHW}.

The Pekar-Tomasevich approximation to the large polaron model of
Fr\"ohlich describes the energy of $N$ polarons through an effective
functional that depends on the wave function $\Psi\in\mathcal{H}_N:= \wedge^N L^2(\R^{3}\times\{1,\ldots,q\})$ of the
particles only. We are mainly interested in the case of spin-$1/2$
fermions but we can allow for arbitrary $q\in\N$ without more
effort. The functional is then given by
\begin{equation}\label{PTF}
    \mathcal{E}^{N,\alpha}(\Psi)=\bigg\langle \Psi, \bigg(\sum_{j=1}^N D_{A,x_j}^2+ \sum_{i<j} \frac{U}{|x_i-x_j|}\bigg)\Psi
\bigg\rangle - \frac{\alpha}{2} \int \frac{\rho_\Psi(x)
\rho_\Psi(y)}{|x-y|} dx dy,
\end{equation}
where $U,\alpha>0$ are constants, and
\begin{equation}\label{rdef}
\rho_{\Psi}(x):=\sum_{j=1}^N \sum_{\sigma_j =1}^q \int
|\Psi(\posp_1,\ldots, \posp_{j-1}, (x, \sigma_j) ,
\posp_{j+1},\ldots, \posp_N)|^2
d\posp_1\ldots\widehat{d\posp_j}\ldots d\posp_N,
\end{equation}
is the density associated to $\Psi$. We have introduced the
notations $\posp_j=(\po_j, \sigma_j)$ for elements of $\R^3\times
\{1,\ldots,q\}$ and we set $ \int d \posp_j=\sum_{\sigma_j=1}^q \int
d \po_j$. Of course in \eqref{rdef} the sum with respect to $j$ may
be replaced by a factor of $N$, due to the symmetry of $\Psi$; but
we shall allow for Boltzons later on, and hence we prefer
\eqref{rdef} as the definition of $\rho_\Psi$. Furthermore,
$D_{A,x}:=-i\nabla+A(x)$ where the vector potential $A:\R^3\to\R^3$
generates a magnetic field $B=\Curl A$. We are primarily interested
in the case where $B$ is constant and hence $A$ will be assumed
linear. The positive parameters $U,\alpha$ are constrained by
$\alpha<U$ due to their role in the Fr\"ohlich large polaron model.
Mathematically, any real values are conceivable for $U$ and
$\alpha$, but $0<U<\alpha$ leads to thermodynamic instability
\cite{GM2010}. The energy of the fields $U^{3N/2}\Psi(U
\po_1,\sigma_1,..., U \po_N, \sigma_N)$ and $UA(Ux)$ upon the
substitutions $Ux\to x$ and $\alpha/U\to \alpha$ becomes
proportional to $U^2$. We therefore set $U=1$ and we require that
$0<\alpha<1$.

It is easy to see, using the diamagnetic and the Hardy inequalities,
that $ \mathcal{E}^{N,\alpha}$ is bounded below if restricted to the
unit sphere $\|\Psi\|=1$. The minimal energy,
\begin{equation}\label{PTUgs}
   E^{N,\alpha}_{PT}:= \inf_{\|\Psi\|=1}
\mathcal{E}^{N,\alpha}(\Psi),
\end{equation}
is therefore finite. By moving particles apart, one can see that
$E_{PT}^{N,\alpha} \leq E_{PT}^{k,\alpha} +E_{PT}^{N-k,\alpha}$ for
$k=1,\ldots,N-1$. The question is, whether it takes energy to do
this, that is, whether for some $\alpha<1$,
\begin{equation}\label{physicalbinding}
    \Delta E^{N,\alpha}_{PT} :=  \min_{1\leq k\leq N-1}\left\{E^{k,\alpha}_{PT} + E^{N-k,\alpha}_{PT}\right\} - E^{N,\alpha}_{PT}  > 0.
\end{equation}
Our main result is the following theorem:

\begin{theorem}\label{main1}
Assume that the vector potential $A$ is linear (constant magnetic
field $B$). Then, for all $N \in \mathbb{N}$ there exists
$\alpha_{N,B}<1$ such that for $\alpha_{N,B}<\alpha<1$ and $U=1$:
\begin{itemize}
\item[(a)] the binding inequality \eqref{physicalbinding} holds,
\item[(b)] the functional \eqref{PTF} has a minimizer.
\end{itemize}
\end{theorem}

Analog results hold in the case of bosons and boltzons, that is, for
$\mathcal{H}_N=\otimes_s^N L^2(\R^{3}\times\{1,\ldots,q\})$, the
symmetric product of $N$ copies of
$L^2(\R^{3}\times\{1,\ldots,q\})$, or $\mathcal{H}_N=\otimes^N
L^2(\R^{3})$ without symmetry requirements. The proofs in these
cases are similar and in the case of Boltzons the proof of (a)
becomes much easier. Yet the property (a) even for boltzons is a
subtle correlation effect since the restriction $\alpha<1$ means
that the Coulomb repulsion dominates the attraction for states of
the form $\vphi_1\otimes\ldots\otimes \vphi_N$.  We remark that 
Theorem~\ref{main1}  has consequences for the binding of boltzonic polarons in the large polaron model of
Fr\"ohlich \cite{AL,GW}.

For $\alpha=0$ there is no minimizer and, in the absence of magnetic
fields, there is no binding for $\alpha$ small enough \cite{FLST}.
The existence of a minimizer is a phenomenon due to the
non-linearity and it occurs whenever the binding inequality
\eqref{physicalbinding} is satisfied (and $\alpha>0$). For other
non-quadratic energy-functionals associated with many-body quantum
systems this has previously been pointed out and described as a
non-linear HVZ-Theorem \cite{Le, Fr}. In this paper we show that
$(a)\Rightarrow (b)$ is a consequence of a \emph{linear} HVZ-Theorem
for an $N$-body Hamiltonian that is intimately related with the
physics of the polaron problem: there is a Hamiltonian $H_\sigma$
depending on a charge density $\sigma\in L^1(\R^3)$ such that
$$\mathcal{E}^{N,\alpha}(\Psi) \leq \langle\Psi,H_{\sigma}\Psi\rangle$$
with equality for $\sigma=\rho_\Psi$. We may think of $\alpha
\sigma$ as the charge density due to a hypothetical, possibly
non-optimal, lattice deformation caused by the electrons. For
$\mathcal{E}^{N,\alpha}(\Psi_n)$ near $E_{PT}^N$, $(\Psi_n)$ being a
minimizing sequence with densities $(\rho_n)$, the binding
inequality implies that $H_{\rho_n}$ has an isolated ground state
energy separated from the essential spectrum of $H_{\rho_n}$ by a
gap that is uniform in $n$ along a subsequence. This uniformity
implies uniform localization of $\Psi_n$ (or concentration of
minimizing sequences) up to magnetic translations.

Our proof of part (a) in Theorem \ref{main1} is based on a
variational argument that is inspired by \cite{GHW} but is
considerably more involved in the present case of particles with
statistics.

The following theorem gives further information about the minimizers
found in Theorem~\ref{main1}. In Theorem~\ref{main2} and throughout
the paper we use the notation $V_\rho := \rho*|\cdot|^{-1}$.

\begin{theorem}\label{main2}
If $\Psi \in \mathcal{H}_N$ is a minimizer of $
\mathcal{E}^{N,\alpha}$, then it solves the non-linear Schr\"odinger
equation
\begin{equation}\label{EL}
  \left(\sum_{k=1}^N (D_{A,x_k}^2-\alpha V_{\rho}(x_k))+\sum_{i<j}
  \frac{1}{|x_i-x_j|}\right) \Psi =\lambda \Psi,
\end{equation}
where $\lambda \in \mathbb{R}$ is the lowest point in the spectrum
of the Schr\"odinger operator on the left hand side and $\rho$ is
the density of $\Psi$. Moreover, if \eqref{physicalbinding} holds
then the spectrum of the  Schr\"odinger operator on the left hand side is
discrete below $\lambda+\Delta E^{N,\alpha}_{PT}$ and hence if 
$\beta\in\R$ with $\beta^2< \Delta E^{N,\alpha}_{PT}$, then
\begin{equation}\label{expodecay1}
e^{\beta|.|} \Psi \in \mathcal{H}_N.
\end{equation}
\end{theorem}

In the case $N=1$, $A=0$ the Pekar-Tomasevich functional reduces to the
Pekar or Choquard functional which is well known to be minimized by
a spherically symmetric, positive function that is unique up to
translations \cite{L,Li}.

Existence of a magnetic polaron and the binding of two polarons subject to an external
magnetic field, not necessarily constant, was previously established
in \cite{GHW}. In the present paper, the methods developed in
\cite{GHW} are extended and generalized to the case of $N>2$
particles of fermionic,
bosonic or bolzonic nature. Results similar to ours in the case $A=0$ where
previously obtained by Lewin in \cite{Le}. Lewin establishes a bound
on the binding energy of the form of a van der Waals potential with
exponentially small corrections. To this end he uses the variational
state introduced by Lieb and Thirring in connection with the van der
Waals binding of neutral atoms and molecules \cite{LieThi1986}. This
approach makes crucial use of spherical averaging and the Newton's
theorem. It brakes down in the presence of a magnetic field where
the rotational invariance of $\mathcal{E}^{N,\alpha}$ is broken. Moreover,
in the absence of a magnetic field our Theorem~\ref{main2} gives more 
information than the corresponding result of Lewin, as it relates the binding
energy $\Delta E^{N,\alpha}_{PT}$ to the gap between $\lambda$ and the
essential spectrum of the Hamiltonian in \eqref{EL}. Lewin, in the case of
binding, merely finds that such a gap exits provided that $\alpha>1-1/N$. 

The Theorem~\ref{main2} opens the following new view upon the phenomenon of $N$-polaron binding: if a Hamiltonian of
the type in \eqref{EL} with total positive charge $\alpha N$ is shown not
to bind $N$ electrons, then binding of $N$ polarons is
excluded. Here binding means positivity of the binding energy. --
In the case where the density is spherically symmetric and $A=0$ we
deduce from \cite{Lieb1984} that the Hamiltonian has no ground state if
$\alpha \leq (1-N^{-1})/2$. This leads to the following corollary: if
$\alpha\leq (1-N^{-1})/2$ then a hypothetical minimizer of $\mathcal{E}^{N,\alpha}$ cannot
have a spherically symmetric density.

This paper is organized as follows: In Section \ref{outline} we
outline the proof of our main Theorem and we introduce the
most important tools. In Section \ref{operineqnovan} we prove an
operator inequality which is of crucial importance for the proof of
existence of a minimizer of the Pekar-Tomasevich functional, as well as the proof of the second part
of Theorem \ref{main2}. In
Section \ref{minexi} we use the operator inequality to prove
existence of a minimizer and exponential decay of any minimizer of
the Pekar-Tomasevich functional. In Section
\ref{bindinginductionproof} we establish the binding inequality
\eqref{physicalbinding}.

\medskip\noindent
\emph{Acknowledgements.} The first author (I.A.) is grateful to
Fabian Hantsch and David Wellig for numerous stimulating discussions and for
introducing him to the theory of multipolarons. He also thanks Mathieu
Lewin for interesting discussions on $N$-body quantum systems.


\section{Preparations and elements of the proofs} \label{outline}

The minimal energy $E_{PT}^{k,\alpha}$ is continuous in $\alpha$
because it is concave in $\alpha$ as the infimum of the affine
functions $\alpha \mapsto \mathcal{E}^{k,\alpha}(\Psi)$. Hence, it
suffices to establish the binding in the case $\alpha=1$.  Our proof
that binding implies existence of a minimizer, i.e $(a) \implies
(b)$ in Theorem \ref{main1}, as well as the proof of Theorem
\ref{main2} readily generalize from the case $\alpha=1$ to any
$\alpha>0$. We therefore put $\alpha=1$ for notational simplicity,
that is,
\begin{equation}\label{PT}
    \mathcal{E}^N(\Psi):= \bigg\langle \Psi, \bigg(\sum_{k=1}^N D_{A,x_k}^2+ \sum_{j<k} \frac{1}{|x_j-x_k|}\bigg)\Psi \bigg\rangle
    - D(\rho_{\Psi}),
\end{equation}
where $D(\rho):=D(\rho,\rho)$,
\begin{equation}\label{ddef}
    D(\rho,\sigma):= \frac{1}{2} \int \frac{\rho(x) \sigma(y)}{|x-y|}\,
    dxdy,
\end{equation}
and
\begin{equation}\label{PTgs}
 E^{N}_{PT}:= \inf_{\|\Psi\|=1} \mathcal{E}^{N}(\Psi).
\end{equation}
The domain of $\mathcal{E}^{N}$ is the form domain, $
\mathcal{Q}_{N,A}$,  of $\sum_{k=1}^N D_{A,x_k}^2$, that is, $
\cQ_{N,A} = \{\Psi\in \cH_N: D_{A,x_k}\Psi\in L^2, \forall k \in
\{1,...,N\}\}$, and we use $\|\cdot\|_{\cQ_{N,A}}$ for the
corresponding form norm. By a \emph{minimizer of} $\mathcal{E}^{N}$
we shall always mean a normalized vector $\Psi\in\cH_N$ with
$\Psi\in\mathcal{Q}_{N,A}$ and $\mathcal{E}^N(\Psi) =
E^{N}_{PT}$. Throughout the paper we
use $\langle\cdot,\cdot\rangle$ and $\|\cdot\|$ for the usual inner
products and norms of $\otimes^N L^2(\R^{3}\times\{1,\ldots,q\})$ and $\mathcal{H}_N$.

By the above explanations it remains to prove the following theorem
in order to establish Theorem \ref{main1}:

\begin{theorem}\label{main}
Assume that the vector potential $A$ is linear. Then,
\begin{itemize}
\item[(a)] there exists a minimizer of $\mathcal{E}^1$,
\item[(b)] if $\mathcal{E}^1,...,\mathcal{E}^{N-1}$ have minimizers then
\begin{equation}\label{binding}
E^{N}_{PT} < E^{k}_{PT} + E^{N-k}_{PT}, \quad \forall k=1,...,N-1,
\end{equation}
\item[(c)] if \eqref{binding} holds then $\mathcal{E}^N$ has a minimizer.
\end{itemize}
\end{theorem}
Part (a) of Theorem \ref{main} is known from \cite{GHW} but we shall
reprove it as a part of the proof of part (c). Part (b) is proved
in Section~\ref{bindinginductionproof} by variational arguments.
Sections \ref{operineqnovan} and \ref{minexi} are devoted to the
proof of (c). The remainder of the present section describes the
difficulties met in the proof of (c) and collects our tools for
dealing with them.

Any proof of (c) must deal with the following translation invariance
of $\mathcal{E}^N$: Let $A: \mathbb{R}^3 \rightarrow \mathbb{R}^3$
be linear, $h \in \mathbb{R}^3$, and $\Psi \in \mathcal{H}_N$. If
$T_h\Psi$ is defined by
\begin{equation}\label{trans}
(T_h \Psi)( x_1,\sigma_1,\ldots,x_N, \sigma_N) =\prod_{j=1}^N e^{i
A(h) \cdot x_j} \Psi(x_1+h,\sigma_1,\ldots,x_N+h, \sigma_N),
\end{equation}
then $\rho_{T_h \Psi}(x)= \rho_{\Psi}(x+h)$ and
\begin{equation}\label{trans1}
  \mathcal{E}^{N} (T_h \Psi)= \mathcal{E}^{N} (\Psi).
\end{equation}
Due to \eqref{trans} and \eqref{trans1} a minimizing sequence of
$\mathcal{E}^N$  may converge to the zero function weakly. On the
other hand in view of Lemma \ref{minbound}, a weak limit
$\Psi\in\mathcal{H}_N$ with $\|\Psi\|=1$ is, indeed, a minimizer of
$\mathcal{E}^N$. Our task is thus to find a minimizing sequence of
$\mathcal{E}^N$ that does not suffer any loss of norm in the limit.
One of our tools to this end is the following form of the
Concentration Compactness Principle \cite{Li}:

\begin{proposition}\label{lions}
Let $(\rho_k)_{k \geq 1}$ be a sequence of nonnegative functions in
$L^1(\mathbb{R}^3)$ with $\int \rho_k=N$. Then there exists a
subsequence of $(\rho_k)$, denoted by $(\rho_k)$ as well, such that
one of the following holds:
\begin{itemize}
\item[(i)](Vanishing) For all $R>0$ we have that $$\lim_{k \rightarrow
\infty} \sup_{y \in \mathbb{R}^3} \int_{B(y,R)} \rho_{k} = 0.$$

\item[(ii)](Dichotomy or compactness) There exists $\lambda \in (0,N]$
such that for all $\varepsilon>0$ there exist $R_{\varepsilon}>0$, a
sequence $y_k=y_k(\varepsilon)$ in  $\mathbb{R}^3$, and a sequence
$P_k=P_k(\eps)$ in $\mathbb{R}$ with $P_k \rightarrow \infty$ as
$k\to \infty$, such that the sequences of functions
\footnote{$\chi_A$ denotes the characteristic function of the set
$A$ and $B(y,R)$ is the ball of radius $R$ centered at $y$ in
$\mathbb{R}^3$.}
\begin{align*}
\rho_{k,1} &:=\rho_{k} \chi_{B(y_k(\varepsilon),R_{\varepsilon})}\\
\rho_{k,2} &:=\rho_{k} \chi_{B(y_k(\varepsilon),P_k(\e))^C}
\end{align*}
 satisfy for $k \geq k_0(\varepsilon)$ the bounds
\begin{equation}\label{dic1}
\|\rho_{k}-\rho_{k,1}-\rho_{k,2}\|_{L^1} \leq \varepsilon,
\end{equation}
\begin{equation}\label{dic2}
|\|\rho_{k,1}\|_{L^1}-\lambda| \leq \varepsilon, \quad
|\|\rho_{k,2}\|_{L^1}-(N-\lambda)| \leq \varepsilon
\end{equation}
and
\begin{equation}\label{distchik}
\text{dist} (\supp \rho_{k,1}, \supp \rho_{k,2}) \rightarrow
\infty,\qquad(k\to\infty).
\end{equation}
If $m$ is a positive integer such that $m \lambda>N$, then after
passing to a subsequence once more, there exists
$\e_1,...,\e_{m-1}>0$, and $\delta>0$ such that
\begin{equation}\label{deltaoverlap}
 \liminf_{k \rightarrow \infty} \int_{\cup_{j=1}^{m-1}
B_{k,\e_j}} \rho_{k,1} \geq \delta
\end{equation}
for all $\e>0$ small enough. Here $B_{k,\e}= B(y_k(\e), R_\e)$.
\end{itemize}
\end{proposition}

\begin{proof}
We shall only prove the last part of (ii). The rest is a variation
of the Concentration Compactness Principle. Let $\e_1, \delta_1 >0$
be such that $m(\lambda-\e_1) > N+ m \delta_1$. Assuming that the
lemma is wrong we inductively construct $\e_1> \e_2>...> \e_m>0$ and
a subsequence of $\rho_k$ denoted by $\rho_k$ as well, such that
\begin{equation*}
 \int_{ \cup_{i=1}^{l-1} B_{k,\e_i} \cap B_{k, \e_l}} \rho_k \leq \frac{\delta_1}{l-1}, \quad \forall l=2,...,m.
\end{equation*}
Using this together with \eqref{dic2} and the inequality
\begin{equation*}
\chi_{\cup_{j=1}^m B_{k, \e_j}} \geq \sum_{j=1}^m \chi_{B_{k,\e_j}}-
\sum_{i<j \leq m} \chi_{B_{k,\e_i} \cap B_{k, \e_j}},
\end{equation*}
we obtain that
\begin{equation*}
\liminf \int_{\cup_{j=1}^m B_{k, \e_j}} \rho_k \geq \sum_{j=1}^m
(\lambda-\e_j) -(m-1)\delta_1 \geq m(\lambda-\e_1)-(m-1) \delta_1>N,
\end{equation*}
where the last inequality follows by the choice of $\e_1$ and
$\delta_1$. This is in contradiction with $\int \rho_k=N$, which
concludes the proof of the lemma.
\end{proof}

The following lemma is the reason for the new part
\eqref{deltaoverlap} in the above version of the Concentration
Compactness Principle.

\begin{lemma}\label{Dgaplemma}
In the case (ii) of Proposition~\ref{lions}, if $(\rho_k)$ is chosen
to satisfy \eqref{deltaoverlap}, then
\begin{equation}\label{Dgap}
   \liminf_{k\to\infty} D(\rho_{k,1})>0
\end{equation}
uniformly for small enough $\e$. (Recall that $\rho_{k,1}$ depends
on $\vep$.)
\end{lemma}

\begin{proof}
Let $\e$ be small enough for \eqref{deltaoverlap} and fixed. By
\eqref{deltaoverlap} there exists $k_0(\e)$ such that for all $k
\geq k_0(\e)$ we have
\begin{equation}\label{deltaovercons}
\int_{\cup_{j=1}^{m-1} B_{k, \e_j}} \rho_{k,1} \geq \frac{(m-1)
\delta}{m}.
\end{equation}
This means that $\int_{B_{k,\e_j}} \rho_{k,1}\geq \delta/m$ for some
$j \in \{1,...,m-1\}$ depending on $k \geq k_0(\e)$. Since
$\text{diam}(B_{k,\e_j})=2 R_{\e_j}$, we conclude that
\begin{equation}
     D(\rho_{k,1}) \geq \frac{1}{2 R_{\e_j}} \left(\int_{B_{k, \e_j}}
     \rho_{k,1}(x) dx \right)^2 \geq \min_i \frac{\delta^2}{2 R_{\e_i} m^2},
\end{equation}
which proves the lemma.
\end{proof}

We want to construct a minimizing sequence $(\Psi_k)$ that is
concentrated near the origin (after translations). Applying the
Concentration Compactness Principle to $|\Psi_k|^2$ would not work,
because the Pekar-Tomesevich functional is invariant under
translations of the form \eqref{trans}, only, and not under general
translations in $\mathbb{R}^{3N}$. Thus, we apply the Concentration
Compactness Principle to the densities, where dichotomy may mean
various things for the wave function. Rather than trying to exclude
all of them we show directly that non-vanishing of the sequence
$\rho_k$, leads to concentration of a subsequence of $\Psi_k$. This
is possible thanks to an HVZ-type operator inequality for the
Hamiltonians $H_{\rho_k}^N$ defined as follows: for a given
real-valued density  $\sigma  \in L^1(\mathbb{R}^3) \cap
L^{6/5}(\mathbb{R}^3)$
we define
\begin{equation}\label{vs}
   V_\sigma:=\sigma*\frac{1}{|.|}
\end{equation}
and
\begin{equation}\label{HsN}
H_\sigma^N := \sum_{j=1}^N (D_{A,x_j}^2-V_{\sigma}(x_j)) +\sum_{i<j}
\frac{1}{|x_i-x_j|}+ D(\sigma),
\end{equation}
which is well defined by the choice of $\sigma$ (\cite{LL} Corollary
5.10). In all the following this operator is considered defined in
$\mathcal{H}_N$ unless explicitly stated otherwise. The following
lemma, taken from  \cite{FLST}, relates the Pekar-Tomasevich
functional to the linear Hamiltonian \eqref{HsN}:


\begin{lemma}[\textbf{Linearization of the Pekar-Tomasevich functional}]\label{Hsin}
For any density $\sigma\in L^1(\R^3)\cap L^{6/5}(\R^3)$,
\begin{equation}\label{PTineq}
      \mathcal{E}^N(\Psi) \leq \langle \Psi,H_{\sigma}^N\Psi \rangle
\end{equation}
with equality if and only if $\sigma=\rho_{\Psi}$. In particular,
for all $N \in \mathbb{N}$,
\begin{equation}\label{Hamineq}
   H_\sigma^N \geq E^N_{PT}.
\end{equation}
In particular, if $(\Psi_k)$ is a minimizing sequence for $\mathcal{E}^N$
and $(\rho_k)$ is the sequence of the corresponding densities, then
\begin{equation}\label{infsigma}
      \lim_{k \rightarrow \infty} \langle \Psi_k, H_{\rho_k}^N \Psi_k\rangle = E_{PT}^N.
\end{equation}
\end{lemma}

\begin{proof}
By the definitions of $H_\sigma^N, \mathcal{E}^N$, $V_\sigma$ and
$D$ we have that
\begin{equation*}
\langle \Psi,H_{\sigma}^N
 \Psi \rangle-\mathcal{E}^N(\Psi) =D(\sigma)+D(\rho_{\Psi})-2 D(\rho_\Psi,
 \sigma)=D(\sigma-\rho_{\Psi}) \geq 0,
\end{equation*}
where the last inequality follows from the positivity of the Fourier
transform of $|.|^{-1}$. This proves \eqref{PTineq}. Inequality
\eqref{Hamineq} follows from \eqref{PTineq} and from the definition,
Equation~\eqref{PTgs}, of $E^N_{PT}$. Equation~\eqref{infsigma}
follows from $\langle \Psi_k,H_{\rho_k}^N\Psi_k \rangle =
\mathcal{E}^N(\Psi_k)$ and from the choice of $(\Psi_k)$.
\end{proof}

The main steps in our proof of part (c) of Theorem~\ref{main} are as
follows:

\medskip
\textbf{Step 1} is to exclude vanishing for the sequence of the
densities $(\rho_k)$ associated with a minimizing sequence
$(\Psi_k)$. To this end we prove that vanishing implies that
$D(\rho_k) \rightarrow 0$ which is easily seen to be in
contradiction with
 $\mathcal{E}^N(\Psi_k) \rightarrow E_{PT}^N$.

As vanishing has now been excluded, the second alternative of
Proposition \ref{lions} must apply to the densities $(\rho_k)$ of
any minimizing sequence $(\Psi_k)$. Upon the translations $\Psi_k
\rightarrow T_{y_k} \Psi_k$, see \eqref{trans}, we may assume that
some part of the densities $\rho_k$ is concentrated near the origin.

\textbf{Step 2} is the proof of the operator inequality
\begin{equation}\label{opineq}
H_{\rho_k}^N \geq
E_{PT}^N+d(1-J_{\varepsilon})+O(\sqrt{\varepsilon}),
\end{equation}
where $d>0$, $J_\varepsilon$ is compactly supported and $0 \leq
J_{\varepsilon} \leq 1$. The proof of \eqref{opineq} is based on the
properties of $\rho_k$ as described by Proposition 2.2 (ii), on
Lemma~\ref{Dgaplemma}, and on a suitable partition of unity that is
adjusted to the supports of $\rho_{k,1}$ and $\rho_{k,2}$.

\textbf{Step 3} is to show that \eqref{opineq} implies concentration
of $(\Psi_k)$.  This is easily done with the help of \eqref{infsigma} and the fact that $\varepsilon$ in
\eqref{opineq} may be taken arbitrarily small.

\section{Absence of vanishing and the operator inequality}
\label{operineqnovan}

Our goal in this Section is to establish absence of vanishing of the
sequence of the densities $(\rho_k)$ associated with a minimizing
sequence $(\Psi_k)$ and to prove the operator inequality of
Proposition~\ref{IMS}.

\begin{lemma}[\textbf{Absence of vanishing}]\label{novanish}
The sequence of the densities $(\rho_k)$ associated with a
minimizing sequence  $(\Psi_k)$ of $\mathcal{E}^N$ cannot be
vanishing.
\end{lemma}

\begin{proof}
We shall derive a contradiction from the assumptions that $(\Psi_k)$
is minimizing and that $(\rho_k)$ is vanishing at the same time. The
vanishing of $(\rho_k)$ implies that
\begin{equation}\label{v1}
      \lim_{k\to\infty} D(\rho_k) = 0,
\end{equation}
as we will prove shortly. By \eqref{PT} and \eqref{v1} we have that
\begin{equation}\label{v2}
    \varliminf_{k\to\infty}\mathcal{E}^N(\Psi_k)  \geq N\inf\sigma(D_A^2) \geq N |B|.
\end{equation}
On the other hand $ E_{PT}^N \leq NE_{PT}^1$ by general principles
and $E_{PT}^1 < |B|$, by \cite{GHW}.
It follows that
\begin{equation}\label{EPTbound}
E_{PT}^N< N |B|,
\end{equation}
which we combine with \eqref{v2} to conclude that the sequence
$(\Psi_k)$ is not minimizing in contradiction to our assumption.

We now turn to the proof of \eqref{v1}. From  $\|\rho_k\|_{L^1}=N$
it follows that, for any $r>0$,
\begin{equation}\label{v5}
    D(\rho_k) \leq \int_{|x-y| \leq r} \frac{\rho_k(x) \rho_k(y)}{|x-y|} dx dy +\frac{N^2}{r}
\end{equation}
and
\begin{equation}\label{v3}
   \int_{|x-y| \leq r} \frac{\rho_k(x) \rho_k(y)}{|x-y|} dx dy \leq N
   \sup_{x \in \mathbb{R}^3}\int_{|x-y| \leq r} \frac{\rho_k(y)}{|x-y|}\,dy.
\end{equation}
For each $x \in \mathbb{R}^3$, by Cauchy-Schwarz,
\begin{equation}\label{v4}
     \int_{|x-y| \leq r} \frac{\rho_k(y)}{|x-y|} dy \leq \left(\int_{|x-y|\leq r} \rho_k(y)dy\right)^{1/2}
\left(\int_{|x-y|\leq r} \frac{\rho_k(y)}{|x-y|^2} dy\right)^{1/2}.
\end{equation}
On the right hand side of \eqref{v4}, the first factor vanishes
uniformly in $x$ in the limit $k\to\infty$, by the assumption that
$(\rho_k)$ is vanishing. The second factor is bounded uniformly in
$x$ because of Lemma \ref{minbound} and the estimate
\begin{equation}
\int \frac{\rho_k(y)}{|x-y|^2} dy = \sum_{j=1}^N\int
   \frac{|\Psi_k(\posp_1,\ldots,\posp_N)|^2}{|\po-\po_j|^2}
   d\posp_1\ldots d\posp_N \leq 4 \|\Psi_k\|_{\mathcal{Q}_{N,A}}^2.\label{Hardiamcons}
\end{equation}
Here we used the Hardy and diamagnetic inequalities. As we have now
shown that \eqref{v4} vanishes uniformly in $x$ in the limit
$k\to\infty$, we conclude, combining \eqref{v5}-\eqref{v4}, that
$D(\rho_k)\to 0$ as $k\to\infty$ because $r>0$ may be chosen
arbitrarily large in \eqref{v5}.
\end{proof}


\begin{proposition}\label{IMS}
Suppose that \eqref{binding} holds and let $(\Psi_k)$ be a
minimizing sequence whose densities $\rho_k=\rho_{\Psi_k}$ have the
properties of Proposition~\ref{lions} (ii).  Then there exists a
subsequence of $\rho_k$, denoted by $\rho_k$ as well, and a positive
number $d>0$ such that for all $\varepsilon>0$ small enough there
exists a function $J_\eps\in C_0^{\infty}(\R^{3N};[0,1])$ symmetric
with
 respect to exchange of particle coordinates, such that for all $k \geq k_0(\eps)$
\begin{equation}\label{Hbound}
     H_{\rho_k}^N \geq E_{PT}^N+d(1-\tau_{y_k}J_\varepsilon)-N
     (2\sqrt{\varepsilon} C+ \varepsilon N )-2^N( \varepsilon N)^2,
\end{equation}
where $y_k=y_k(\e)$ is given by Proposition \ref{lions} (ii),
$\tau_{y_k}
J_{\varepsilon}(x_1,...,x_N):=J_{\varepsilon}(x_1-y_k,...,x_N-y_k)$
and $C:=2 \sup \|\Psi_{k}\|_{\cQ_{N,A}}<\infty$ (see Lemma
\ref{minbound}). If the sequence $(\rho_k)$ is concentrated, i.e. if
$\lambda=N$ in Proposition \ref{lions} (ii), then we may choose
$d=\Delta E^N :=\min\{E_{PT}^k+E_{PT}^{N-k}\mid k=1,...,N-1\}-E_{PT}^N$.
\end{proposition}

We fix $\varepsilon>0$ and $(\Psi_k)$ as described in Proposition
\ref{IMS}. Let $(y_k)$ be the corresponding sequence provided by
Proposition \ref{lions} (ii). After the translations $\Psi_k \mapsto
T_{y_k} \Psi_k$ defined by Equation~\eqref{trans} we may assume that
the densities of $(\Psi_k)$ have the properties of Proposition
\ref{lions} (ii) with $y_k=0$. It thus remains to prove Proposition
\ref{IMS} in the case $y_k=0$. As a preparation we will first
establish the following two lemmas.


\begin{lemma}[\textbf{Partition of unity}]\label{partun}
 Let $\varepsilon$ and $\Psi_k$ be as explained above. Let also
 $\rho_k=\rho_{\Psi_k}$ and $\rho_{k,i}$ be as in Proposition
 \ref{lions} (ii). Then there exist $k_0 \geq 1$ and non-negative functions $j_{1}, j_{2}:
\mathbb{R}^3 \rightarrow \mathbb{R}$ with
\begin{equation}\label{prepartunprop}
0 \leq j_1, j_2 \leq 1, \quad j_{1}^2 +j_{2}^2=1, \quad \|\nabla
j_{i}\|_{L^\infty} \leq  \varepsilon, \quad \supp j_1 \subset
B(0,R_\varepsilon+\frac{3}{\varepsilon}),
\end{equation}
 such that for all $k \geq k_0$,
\begin{equation}\label{partundist}
\text{dist}(\supp \rho_{k,i}, \supp j_{3-i}) \geq
\frac{1}{\varepsilon}, \quad i=1,2.
\end{equation}
If $a=(a_1,...,a_N) \in \{1,2\}^N$, then the functions
\begin{equation}\label{partundef}
J_{a}(x_1,...,x_N):= \prod_{j=1}^N j_{{a_j}}(x_j)
\end{equation}
have the following properties:
\begin{equation}\label{partunprop}
0 \leq J_{a} \leq 1, \quad \sum_{a \in \{1,2\}^N} J_{a}^2=1, \quad
\|\nabla J_{a}\|_{L^\infty} \leq  \varepsilon N.
\end{equation}
\end{lemma}

\begin{proof}
It is an elementary exercise to construct non-negative functions
$f_1,f_2\in C^{\infty}(\R)$ with $\sup_x|f_{\ell}'(x)|\leq 1$,
$f_1^2+f_2^2=1$, $f_1=1$ on $(-\infty,1]$ and $f_2=1$ on
$[3,\infty)$. Let
\begin{equation*}
     j_\ell(x)=f_\ell((|x|-R_\varepsilon) \varepsilon).
\end{equation*}
Using the properties of $f_1,f_2$ and the fact that $P_k(\e)
\geq R_{\eps}+4\eps^{-1}$ for $k$ large enough, see Proposition
\ref{lions} (ii), one easily verifies that $j_1, j_2$ have the
desired properties. \eqref{partunprop} follows from
\eqref{partundef} and the properties of $j_1, j_2$.
\end{proof}


\begin{lemma}\label{Vcontrol}
Let $\varepsilon$ and  $(\Psi_k)$ be as in Lemma \ref{partun}, and
$C:=2 \sup \|\Psi_{k}\|_{\cQ_{N,A}}$ as in Proposition \ref{IMS}. If
$\rho_{k}, \rho_{k,i}$ are given by Proposition~\ref{lions} (ii),
then for $k$ large enough,
\begin{eqnarray}
    V_{\rho_k}-V_{\rho_{k,1}}-V_{\rho_{k,2}} &\leq & \sqrt{\eps}C, \label{dichcons2}\\
   (V_{\rho_k}-V_{\rho_{k,i}}) j_{i}^2 &\leq& (\sqrt{\eps} C+\varepsilon N )j_{i}^2, \quad i=1,2. \label{vrhokiappr}
\end{eqnarray}
\end{lemma}

\begin{proof}
By the definitions of $V_{\rho_k}$,  $V_{\rho_{k,1}}$, and
$V_{\rho_{k,2}}$, we have
$$
V_{\rho_k}-V_{\rho_{k,1}}-V_{\rho_{k,2}}=(\rho_k-\rho_{k,1}-\rho_{k,2})*\frac{1}{|.|},
$$
where $0\leq \rho_k-\rho_{k,1}-\rho_{k,2} \leq  \rho_k$. Hence, by
Cauchy-Schwarz, \eqref{Hardiamcons}, and \eqref{dic1},
\begin{align*}
|(V_{\rho_k}-V_{\rho_{k,1}}-V_{\rho_{k,2}})(x)| &\leq \left(\int
\frac{\rho_k(y)}{|x-y|^2} dy \right)^{1/2} \left(\int(\rho_k - \rho_{k,1}-\rho_{k,2} ) dy \right)^{1/2}\\
    &\leq C\sqrt{\e}.
\end{align*}

 To prove \eqref{vrhokiappr}, by \eqref{dichcons2} it suffices to show that
$V_{\rho_{k,3-i}} j_{i}^2 \leq \varepsilon N j_{i}^2$. This easily
follows from \eqref{partundist} and $\|\rho_{k,3-i}\|_{L^1} \leq
\|\rho_k \|_{L^1} = N$.
\end{proof}


\begin{proof}[\textbf{Proof of Proposition~\ref{IMS}}]
In this proof we shall tacitly assume that $k$ is large enough so
that the statements of the previous lemmas apply. By the IMS
localization formula \cite{CFKS},
\begin{equation}\label{IMSbound1}
H_{\rho_k}^N=\sum_{a \in \{1,2\}^N} J_{a} H_{\rho_k}^N J_{a} -
\sum_{a \in \{1,2\}^N}|\nabla J_{a}|^2.
\end{equation}
We will now estimate the terms $J_{a} H_{\rho_k}^N J_{a}$ from
below.

\underline{1st Case}: $a$ has $n$ ones and $N-n$ twos, $0 < n < N$.
We may assume without loss of generality that $a=(1,...,1,2,...,2)$.
From $\rho_k \geq \rho_{k,1}+\rho_{k,2}$ it follows that
\begin{equation}\label{dichcons1}
D(\rho_k) \geq D(\rho_{k,1})+D(\rho_{k,2}).
\end{equation}
This, together with \eqref{vrhokiappr} and \eqref{Hamineq} implies
that
\begin{align}
J_{a} H_{\rho_k}^N J_{a} &\geq
J_{a}(H_{\rho_{k,1}}^n+H_{\rho_{k,2}}^{N-n})
J_{a}-N(\sqrt{\varepsilon}C+\eps
N )J_{a}^2,\nonumber\\
    &\geq  (E_{PT}^n+E_{PT}^{N-n})J_{a}^2-N(\sqrt{\varepsilon}C+ \varepsilon N )J_{a}^2. \label{IMSbound3}
\end{align}
Note that $H_{\rho_{k,1}}^n$ acts on the coordinates labeled by
$1,\ldots,n$, while $H_{\rho_{k,2}}^{N-n}$ acts on the ones labeled by
 $n+1,\ldots,N$. Moreover, $J_a\mathcal{H}_N\subset \mathcal{H}_n \otimes
\mathcal{H}_{N-n}$ by construction of $J_a$.

\underline{2nd case}: $a=(2,\ldots,2)$, i.e., only twos. By
\eqref{vrhokiappr} and \eqref{dichcons1},
\begin{equation}\label{sec}
J_{a} H_{\rho_k}^N J_{a} \geq J_{a} (D(\rho_{k,1}) +
H_{\rho_{k,2}}^N) J_{a}-N(\sqrt{\eps} C+ \eps N )J_{a}^2.
\end{equation}
By \eqref{Hamineq} we have $H_{\rho_{k,2}}^N \geq E_{PT}^N$ and by
Lemma~\ref{Dgaplemma} there exits a constant $\gamma>0$ such that
\begin{equation}\label{cuni}
     D(\rho_{k,1}) \geq \gamma,\quad\text{for}\ \e\ \text{small enough}.
\end{equation}
It follows that, for $\e$ small enough,
\begin{equation}\label{IMSbound2}
J_{a} H_{\rho_k}^N J_{a} \geq (E_{PT}^N+\gamma)
J_{a}^2-N(\sqrt{\eps} C+ \eps N) J_{a}^2.
\end{equation}

\underline{3rd case}: $a=a_0:=(1,\ldots,1)$. Since $H_{\rho_{k}}^N
\geq E_{PT}^N$, we have
\begin{equation}\label{IMSbound4}
    J_{{a_0}} H_{\rho_k}^N J_{{a_0}} \geq E_{PT}^N J_{{a_0}}^2.
\end{equation}

Combining the results \eqref{IMSbound3}, \eqref{IMSbound2} and
\eqref{IMSbound4} from the three cases above with \eqref{partunprop}
and \eqref{IMSbound1} we obtain \eqref{Hbound} with
$J_\varepsilon=J_{a_0}^2$ and $d=\min\{\gamma, \Delta E^N\}$, which
is positive due to the binding assumption~\eqref{binding}.

In the case $\lambda=N$ we may improve our bound in the second case
to get $d= \Delta E^N$. Indeed
\begin{equation}
H_{\rho_{k,2}}^N \geq \sum_{j=1}^N
\left(D_{A,x_j}^2-V_{\rho_{k,2}}(x_j)\right) \geq NE_{PT}^1-N
C\sqrt{\vep},
\end{equation}
because $D_{A,x_j}^2 \geq E_{PT}^1$ and $V_{\rho_{k,2}}(x) \leq
C\|\rho_{k,2}\|_{L^1}^{1/2} \leq C\sqrt{\eps}$ by the Cauchy-Schwarz,
 Hardy and diamagnetic inequalities. Here we used $\lambda=N$ and
\eqref{dic2}. Since $NE_{PT}^1 \geq E_{PT}^1+E_{PT}^{N-1}$ we
conclude that
$$
    J_{a} H_{\rho_k}^N J_{a} \geq (E_{PT}^N + \Delta E^N) J_{a}^2-N(\sqrt{\e} 2C+ \e N) J_{a}^2,
$$
which we use in place of \eqref{IMSbound2}.
\end{proof}


\section{Existence of a minimizer and exponential decay}\label{minexi}

 In this Section we prove parts (a),(c) of Theorem \ref{main} and then we prove Theorem \ref{main2}. The
 part (b) of Theorem \ref{main} will be proved in the next Section.

\begin{lemma}\label{nodicho}
Assume that \eqref{binding} holds. Then, there exists a minimizing
sequence $(\Phi_k)$ with the following property: for every
$\delta>0$ there exists $P>0$ such that
\begin{equation}\label{concentration}
\liminf_{k \rightarrow \infty} \int_{B(0,P)} | \Phi_{k}|^2 \geq
1-\delta.
\end{equation}
\end{lemma}

\begin{proof}
Without loss of generality we may assume that $\delta<1/2$. By Lemma
\ref{novanish} there exists a minimizing sequence $(\Psi_k)$ for
which the sequence $(\rho_k)$ of the associated densities satisfies
the properties of Proposition 2.2 (ii) and hence Proposition
\ref{IMS} applies to $(\Psi_k)$. The operator inequality
\eqref{Hbound} implies that
\begin{equation*}
\langle  \Psi_k, H_{\rho_k}^N  \Psi_k \rangle \geq
E_{PT}^N+d-d\langle  \Psi_k, \tau_{y_k} J_\varepsilon
  \Psi_k\rangle-N ( 2 \sqrt{\varepsilon} C+
\varepsilon N)-2^N( \varepsilon N)^2.
\end{equation*}
Upon  rearranging this inequality, it follows from \eqref{infsigma}
that
\begin{equation*}
\liminf_{k \rightarrow \infty} \langle \Psi_k, \tau_{y_k}J_\vep
\Psi_k \rangle \geq 1-\frac{N}{d}( 2 C \sqrt{\vep}+\vep
N)-\frac{2^N}{d}(\vep N)^2 \geq 1-\delta,
\end{equation*}
for $\vep$ small enough. Since $J_\vep$ is compactly supported and
$0 \leq J_\vep\leq 1$ it follows that
\begin{equation}\label{conc}
\liminf_{k \rightarrow \infty} \int_{B(y_k,R)} |\Psi_k|^2 \geq 1
-\delta,
\end{equation}
where $R$ and $y_k$ depend on $\eps$ and hence on $\delta$. Using an
argument of Lions (see \cite{Li})  we shall now replace $(y_k)$ by
an other sequence $(y_k')$ that is independent of $\delta$ such that
\eqref{conc} still holds after enlarging $R$. Let $R'$ and $(y_k')$
be determined in the same way as $R$ and $(y_k)$ in the case
$\delta=1/2$. That is,
\begin{equation*}
   \liminf_{k \rightarrow \infty} \int_{B(y'_k,R')} |\Psi_k|^2 \geq \frac{1}{2}.
\end{equation*}
Since $\|\Psi_k\|=1$ and since $1-\delta>1/2$, by assumption, the
balls $B(y_k,R)$ and $B(y_k',R')$ must overlap for $k$ large enough.
It follows that
\begin{equation}
   \liminf_{k \rightarrow \infty} \int_{B(y'_k,R'+2R)} |\Psi_k|^2 \geq 1 -\delta.
\end{equation}
The sequence $\Phi_k = T_{y'_k}\Psi_k$ is minimizing and it
satisfies \eqref{concentration} with $P=R'+2R$.
\end{proof}

\begin{proof}[\textbf{Proof of Theorem \ref{main} (a), (c) (existence of a minimizer)}]
Let $(\Phi_k)$ be given by Lemma \ref{nodicho}. By Lemma
\ref{minbound}, part (b), $(\Phi_k)$ is  bounded in $\cQ_{N,A}$ and
hence, after passing to a subsequence, we may assume that $\Phi_k
\rightarrow\Phi\in \cQ_{N,A}$ weakly in $\cQ_{N,A}$. Since $A$ is
locally bounded it follows that $\Phi_k \rightarrow \Phi$ locally in
$\mathcal{H}_N$ and weakly in $\mathcal{H}_N$. Hence, by Lemma
\ref{nodicho}, for every $\delta>0$ there exists $P>0$ such that
\begin{equation*}
1=\lim_{k \to\infty} \|\Phi_k\|^2 \geq \|\Phi\|^2\geq
\int_{B(0,P)}|\Phi|^2 dx =\liminf_{k\to\infty}
\int_{B(0,P)}|\Phi_k|^2 dx \geq 1-\delta.
\end{equation*}
It follows that $\|\Phi\|=1$ and hence that $\Phi_k \rightarrow
\Phi$ strongly in $\mathcal{H}^N$. Since $\Phi_k \rightarrow\Phi\in
\cQ_{N,A}$ weakly in $\cQ_{N,A}$, the  parts (a) and (c) of Theorem
\ref{main} follow from Lemma~\ref{minbound}, (c).
\end{proof}


\begin{proof}[\textbf{Proof of Theorem \ref{main2}}]
This proof is based on Lemma~\ref{Hsin}, which clearly holds for any
$\alpha>0$. Let $\Psi$ be a minimizer with density $\rho$. By
Lemma~\ref{Hsin}, $H_\rho^N \geq E_{PT}^N$ and $\langle\Psi,
H_\rho^N\Psi\rangle = E_{PT}^N$. It follows that $\Psi$ belongs to
the domain of the Friedrichs' extension of $ H_\rho^N$ and that
$H_\rho^N\Psi = E_{PT}^N\Psi$. This equation agrees with the
Schr\"odinger equation~\eqref{EL} upon subtracting $D(\rho)\Psi$
from both sides.

By \cite{AG} eigenvalues of $H_{\rho}^N$ below
\begin{align*}
\Sigma &:=\lim_{R \rightarrow \infty} \Big(\inf_{\Phi \in D_R,
\|\Phi\|=1} \langle \Phi, H_{\rho}^N \Phi\rangle\Big),\\
D_R &:=\{\Phi\in\cQ_{N,A} \mid \Phi(x)=0 \text{ for } |x|<R\},
\end{align*}
are associated with exponentially decaying eigenfunctions. This
means that $e^{\beta |.|} \Psi \in L^2$ provided
$\beta^2<\Sigma-E_{PT}^N$. Applying Proposition \ref{IMS} to the
constant minimizing sequence $\Psi_k=\Psi$, for which the sequence
of densities $\rho_k=\rho$ obviously is concentrated, we see that
\begin{equation}\label{operbet}
H_\rho^N \geq E_{PT}^N+ \Delta E^N(1-J_\vep)-O(\sqrt{\varepsilon}),
\end{equation}
where $J_\vep$ is compactly supported and $\vep$ is small enough.
Since $\e$ can be arbitratilly small we obtain that 
$\Sigma \geq E_{PT}^N+\Delta E^N$, which concludes the proof.
\end{proof}

\section{Proof of Binding}\label{bindinginductionproof}

In this Section we prove Theorem \ref{main} part (b). To explain the
main ideas in their pure form, without the difficulties due to the
Pauli-principle, we first do the proof in the case of Bolzons, i.e.,
for Pekar-Tomasevich functional defined on $L^2(\mathbb{R}^{3N})$.
Thereafter we shall describe the modifications necessary to
accommodate fermions and bosons.

\smallskip\noindent
\emph{The case of Boltzons.} The functionals
$\mathcal{E}^1,\ldots,\mathcal{E}^{N-1}$ have minimizers
$\Phi_1,\ldots,\Phi_{N-1}$ by assumption. Assuming that
\begin{equation}\label{keinebin}
  E_{PT}^N=E_{PT}^k+E_{PT}^{N-k}
\end{equation}
for some $k \in \{1,...,N-1\}$ we shall prove in the Steps 1 and 2
below, that on the one hand $\Phi_k \otimes \Phi_{N-k}$ is a
minimizer of $\mathcal{E}^N$, on the other hand it cannot satisfy
the corresponding Euler-Lagrange equation. Hence the assumption
\eqref{keinebin} must be wrong.

\smallskip\noindent
\textbf{Step 1}: $\Phi_k \otimes \Phi_{N-k}$ is a minimizer of
$\mathcal{E}^N$, that is
\begin{equation}\label{tensoradditivitat}
\mathcal{E}^N(\Phi_k \otimes \Phi_{N-k}) =E_{PT}^N.
\end{equation}

From the definitions of the density, $\rho_\Phi$, and interaction
energy $D(\rho_\Phi)$ associated with any $\Psi$ (see \eqref{rdef},
\eqref{ddef}), we easily see that
\begin{equation}\label{tensum}
\rho_{\Phi_k \otimes \Phi_{N-k}}=\rho_{\Phi_k}+\rho_{\Phi_{N-k}}
\end{equation}
and
\begin{equation}\label{dtensor}
D(\rho_{\Phi_k \otimes \Phi_{N-k}})
=D(\rho_{\Phi_k})+D(\rho_{\Phi_{N-k}})+2D(\rho_{\Phi_k},\rho_{\Phi_{N-k}}),
\end{equation}
where
\begin{equation}\label{dzweidichten}
2D(\rho_{\Phi_k},\rho_{\Phi_{N-k}})=\langle \Phi_k \otimes
\Phi_{N-k}, \sum_{i=1}^k \sum_{j=k+1}^N \frac{1}{|x_i-x_j|} \Phi_k
\otimes \Phi_{N-k} \rangle.
\end{equation}
From \eqref{dtensor}, \eqref{dzweidichten}, and the assumption
\eqref{keinebin} it follows that
\begin{align*}
\mathcal{E}^N(\Phi_k \otimes \Phi_{N-k}) &=\mathcal{E}^k(\Phi_k)+
\mathcal{E}^{N-k}(\Phi_{N-k})\\
   &= E_{PT}^k+E_{PT}^{N-k} = E_{PT}^N.
\end{align*}

\smallskip\noindent
\textbf{Step 2}:  $\Phi_k \otimes \Phi_{N-k}$ does not solve the
Euler Lagrange equation of $\mathcal{E}^N$.

Suppose that $\Phi_k\otimes \Phi_{N-k}$ solves the Euler-Lagrange equation
\begin{equation}\label{EL1}
\left(\sum_{j=1}^N D_{A,x_j}^2+\sum_{i<j}^{N}
\frac{1}{|x_i-x_j|}-\sum_{j=1}^N V_{\rho_{\Phi_k \otimes
\Phi_{N-k}}}(x_j) -\lambda \right) \Phi_k \otimes \Phi_{N-k}=0,
\end{equation}
for some $\lambda \in \mathbb{R}$. Since $\Phi_k$ and $\Phi_{N-k}$
are minimizers of $\mathcal{E}^k$ and $\mathcal{E}^{N-k}$,
respectively, they satisfy the Euler-Lagrange equations
\begin{equation}\label{EL2}
\left(\sum_{j=1}^k D_{A,x_j}^2+\sum_{i<j}^{k}
\frac{1}{|x_i-x_j|}-\sum_{j=1}^k V_{\rho_{\Phi_k}}(x_j) -\lambda_1
\right)\Phi_k=0,
\end{equation}
and
\begin{equation}\label{EL3}
\left(\sum_{j=k+1}^N D_{A,x_j}^2+\sum_{k+1\leq i<j}^{N}
\frac{1}{|x_i-x_j|}-\sum_{j=k+1}^N
V_{\rho_{\Phi_{N-k}}}(x_j)-\lambda_2 \right) \Phi_{N-k}=0,
\end{equation}
with $\lambda_1,\lambda_2 \in \mathbb{R}$. Note that, by
\eqref{tensum},
\begin{equation}\label{Vadditivitat}
V_{\rho_{\Phi_k \otimes
\Phi_{N-k}}}=V_{\rho_{\Phi_k}}+V_{\rho_{\Phi_{N-k}}}.
\end{equation}
Taking tensor products of the Equations~\eqref{EL2} and \eqref{EL3}
with $\Phi_{N-k}$ and $\Phi_k$, respectively, and subtracting the
resulting equations from \eqref{EL1}, we obtain that
\begin{equation}\label{eltogether}
\left(\sum_{i=1}^k \sum_{j=k+1}^N \frac{1}{|x_i-x_j|}-\sum_{i=1}^k
V_{\rho_{\Phi_{N-k}}}(x_i) -\sum_{j=k+1}^N V_{\rho_{\Phi_{k}}}(x_j)
-\lambda+\lambda_1+\lambda_2 \right) \Phi_k \otimes \Phi_{N-k}=0.
\end{equation}
Since $V_{\rho_{\Phi_k}}$ and $V_{\rho_{\Phi_{N-k}}}$  are bounded
functions (see \eqref{vrkbound} in the Appendix) the expression in
parentheses is a multiplication operator that is bounded below by
\begin{equation}\label{schrankebel}
\sum_{i=1}^k \sum_{j=k+1}^N \frac{1}{|x_i-x_j|}-M,
\end{equation}
for some $M>0$. Clearly, \eqref{schrankebel} is positive, e.g., for
$x_1$ close to $x_{k+1}$. We may thus find balls $B_1 \subset
\mathbb{R}^{3k}$ and $B_2 \subset \mathbb{R}^{3(N-k)}$ such that
\eqref{schrankebel} is strictly positive on $B_1 \times B_2$. At the
same time we may assume, after suitable magnetic translations of
$\Phi_k, \Phi_{N-k}$, that
\begin{equation}\label{posint}
\int_{B_1 \times B_2} |\Phi_k \otimes \Phi_{N-k}|^2 >0.
\end{equation}
The strict positivity of the lower bound \eqref{schrankebel} and and
the inequality \eqref{posint} are in contradiction with
\eqref{eltogether}, which completes the proof of Step~2.

\medskip\noindent
\emph{The case of fermions.} In the case of fermions, the tensor
product $\Phi_k \otimes \Phi_{N-k}$ of the minimizers $\Phi_k$ and
$\Phi_{N-k}$ in Step 1 must be antisymmetrized and normalized. The
density of the resulting $N$-particle state is not the sum of the
densities of  $\Phi_k, \Phi_{N-k}$. In order to regain an analogue
of  \eqref{tensum} we shall apply smooth space cut-offs at distance
$R$ from the origin and then move $\Phi_{N-k}$ by a distance of
$3R$. These cut-off minimizers, as well as their antisymmetrized
tensor product, are approximate minimizers satisfying approximate
Euler-Lagrange equations, the error being exponentially small. But
such an exponentially small error is not compatible with the power
laws decay of the Coulomb interaction between the first $k$ and the
last $N-k$ particles.

We now proceed with the details. We use $c,C$ to denote positive
constants possibly changing from one equation to another. Suppose
that for some $k$
\begin{equation}\label{annahme}
    E_{PT}^N=E_{PT}^k+E_{PT}^{N-k},
\end{equation}
and let $\psi_m$ be a minimizer of $\mathcal{E}^m$, $m\in
\{1,\ldots,N-1\}$. Let $f\in C^{\infty}(\R;[0,1])$ with $f(s)=1$ if
$s\leq -1$ and  $f(s)=0$ if $s\geq 0$, and let $\chi_R(x) :=
f(|x|-R)$, a smoothed characteristic function of the ball
$B(0,R)\subset \R^3$. We define
$$
     \phi_m=\frac{\psi_m \chi_R^{\otimes m}}{\|\psi_m \chi_R^{\otimes m}\|}.
$$
Let $y \in \mathbb{R}^3$ with $|y|=3R$. Recall that $T_y \phi_{N-k}$
denotes a magnetic translation of $\phi_{N-k}$ as defined in
\eqref{trans}. Due to the exponential decay of the minimizers
$\psi_k, \psi_{N-k}$ and their gradients and Laplacians we obtain
that $\phi_k,T_y\phi_{N-k}$ are approximate minimizers and satisfy
respectively the Euler-Lagrange equations of $\mathcal{E}^k,
\mathcal{E}^{N-k}$ up to an exponentially small error. More
precisely,
\begin{align}
   (H_{\rho_{\phi_k}}^k - E_{PT}^k ) \phi_k &=O_{\cQ_{k,A}}(e^{-cR}),\label{fastELik} \\
   (H_{\rho_{T_y \phi_{N-k}}}^{N-k} -E_{PT}^{N-k}) T_y \phi_{N-k}&= O_{\cQ_{N-k,A}}(e^{-cR}),\label{fastELiNmink}
\end{align}
where $O_{\cQ_{m,A}}$ refers to the $\cQ_{m,A}$ norm. Since
$\mathcal{E}^{m}(\phi)=\langle\phi,H_{\rho_{\phi}}^m\phi\rangle$ it follows that
\begin{equation}\label{ungefaehrmin}
\begin{split}
   \mathcal{E}^{k}(\phi_k) &=E_{PT}^k+ O(e^{-cR}),\\
    \mathcal{E}^{N-k}(T_y\phi_{N-k}) &=E_{PT}^{N-k}+ O(e^{-cR}).
\end{split}
\end{equation}
Equations \eqref{fastELik} and \eqref{fastELiNmink} correspond to
\eqref{EL2} and \eqref{EL3} in the boltzonic case, note however the
irrelevant constants $D(\rho_{\phi_k})$ and $D(\rho_{T_y \phi_{N-k}})$
in the Hamiltonians defined by \eqref{HsN}.

Let now $\Phi:= P_k (\phi_k \otimes T_y \phi_{N-k})$. Here $P_k:=
\sqrt{N \choose k} P_A$ where $P_A$
 denotes the projection onto the completely antisymmetric functions with respect to permutations
 of pairs of positions and spins. The factor in front of $P_A$ is chosen so that $\Phi$ is also normalized.
Since the densities of $\phi_k, T_y \phi_{N-k}$ have disjoint
supports we obtain that $\rho_{\Phi}=\rho_{\phi_k}+\rho_{ T_y
\phi_{N-k}}$ which similarly to the case of Boltzons implies that
\begin{equation}\label{cutadd}
  \mathcal{E}^{N}(\Phi)=\mathcal{E}^{k}(\phi_k)+\mathcal{E}^{N-k}( T_y \phi_{N-k})
\end{equation}
and that
\begin{equation}\label{Vadd}
V_{\rho_{\Phi}}=V_{\rho_{\phi_k}}+V_{\rho_{ T_y \phi_{N-k}}}.
\end{equation}
From \eqref{ungefaehrmin} and \eqref{cutadd} we obtain that
\begin{equation}\label{fastmin}
\mathcal{E}^N(\Phi)= E_{PT}^k +E_{PT}^{N-k} +O(e^{-cR}).
\end{equation}
We show now that $\Phi$ satisfies an approximate Euler Lagrange
equation. We take the tensor product of both sides of
\eqref{fastELik} with $T_y \phi_{N-k}$. Similarly, we take tensor
product of both sides of \eqref{fastELiNmink} with $\phi_{k}$. By
adding the resulting equations and then adding $J_k(\phi_k \otimes
T_y \phi_{N-k})$ on both sides, where
\begin{equation}\label{Jkdef}
J_k(x_1,...,x_N): = \sum_{i=1}^k \sum_{j=k+1}^N \frac{1}{|x_i-x_j|}
-\sum_{i=1}^k V_{\rho_{ T_y \phi_{N-k}}}(x_i) -\sum_{j=k+1}^N
V_{\rho_{\phi_k}}(x_j),
\end{equation}
we arrive at
\begin{equation}\label{bound4pre}
(H_{\rho_{\Phi}}^N -c_R) \phi_k \otimes T_y \phi_{N-k} = J_k(\phi_k
\otimes T_y \phi_{N-k}) + O_{\cQ_{N,A}}(e^{-cR}),
\end{equation}
where $c_R:=E_{PT}^k +E_{PT}^{N-k}-2 D(\rho_{\phi_k}, \rho_{T_y
\phi_{N-k}})$  depends on $R$. We have used \eqref{HsN} and
\eqref{Vadd}. The fact that the supports of $\phi_k, T_y \phi_{N-k}$
have distance $R$ in each particle coordinate implies that
\begin{equation}\label{Jkbound}
  |J_k|=O(R^{-1}), \quad |\nabla J_k| =O(R^{-2}), \text{ uniformly in  } (x_1,...,x_N) \in \supp \phi_k \otimes T_y \phi_{N-k}.
\end{equation}
  Applying the antisymmetrization $P_k$ to both sides of
\eqref{bound4pre} and using \eqref{Jkbound} as well as the symmetry
of $H_{\rho_{\Phi}}^N$ with respect to the $N$ particles, we arrive
at
\begin{equation}\label{bound4}
     \|(H_{\rho_{\Phi}}^N -c_R) \Phi \|_{\cQ_{N,A}} = O(R^{-1}).
\end{equation}
We are now going to improve this error estimate by changing the
Lagrange multiplier by $O(R^{-1})$. To this end we write
 \begin{equation}\label{bound2}
H_{\rho_{\Phi}}^N \Phi=\lambda_R \Phi+ f_R,\ \text{with}\ \langle f_R,\Phi\rangle=0.
\end{equation}
First observe that \eqref{bound4} and \eqref{bound2} imply that
$\lambda_R=c_R+O(R^{-1})$ and therefore
\begin{equation}\label{bound3}
\|f_R\|_{\cQ_{N,A}} = O(R^{-1}).
\end{equation}
On the other hand using \eqref{bound2} twice we obtain that
\begin{equation*}
\langle f_R, f_R \rangle =\langle f_R, H_{\rho_{\Phi}}^N \Phi
\rangle=\langle f_R, (H_{\rho_{\Phi}}^N-E_{PT}^N) \Phi \rangle.
 \end{equation*}
Recall that $(H_{\rho_{\Phi}}^N-E_{PT}^N)|_{\cH_N} \geq 0$ (see Lemma \ref{Hsin}) and $f_R,
\Phi \in \cH_N$. Hence, by Cauchy-Schwarz for positive
(semi-)definite quadratic forms we find
\begin{align*}
    \|f_R\|^2 & \leq \langle f_R, (H_{\rho_{\Phi}}^N-E_{PT}^N) f_R \rangle^{1/2}
\langle \Phi, (H_{\rho_{\Phi}}^N-E_{PT}^N) \Phi \rangle^{1/2}\\ &
\leq c \| f_R \|_{\cQ_{N,A}} \langle \Phi,
(H_{\rho_{\Phi}}^N-E_{PT}^N) \Phi \rangle^{1/2}.
 \end{align*}
This estimate together with $\mathcal{E}^{(N)}(\Phi)= \langle \Phi,
 H_{\rho_{\Phi}}^N \Phi \rangle$ (see Lemma \ref{Hsin}),
 \eqref{annahme}, \eqref{fastmin} and \eqref{bound3} implies that
\begin{equation}\label{chirexp}
  \|(H_{\rho_{\Phi}}^N - \lambda_R) \Phi \| = \| f_R\| = O( e^{-c R}).
\end{equation}
By definition of $\Phi$, the equations \eqref{bound2} and
\eqref{chirexp} imply that
\begin{equation}\label{fastEL}
     \|(H_{\rho_{\Phi}}^N-\lambda_R) \phi_k \otimes T_y \phi_{N-k}\| =O(e^{-cR}),
\end{equation}
because $P_k$ acts isometrically on the left hand side of
\eqref{fastEL} and commutes with $H_{\rho_{\Phi}}^N$. From
\eqref{bound4pre} and \eqref{fastEL} it follows that
\begin{equation}\label{expocon}
\|( J_k+c_R-\lambda_R )  \phi_k \otimes T_y \phi_{N-k}\|=O(e^{-cR}).
\end{equation}
This is in contradiction with Lemma \ref{noexpocon} below. Hence,
our assumption \eqref{annahme} must be wrong and \eqref{binding} is
proved.


\begin{lemma}\label{noexpocon}
If the minimizers $\psi_m, m \in \{k,N-k\}$ are chosen so that
$\int u \rho_{\psi_m}(u) du=0$, then there exists a constant $C>0$ such that
   \begin{equation*}
   \inf_{M \in \mathbb{R}}\|( J_k+M )  \phi_k \otimes T_y \phi_{N-k}\| \geq \frac{C}{R^3}
\end{equation*}
(recall that $|y|=3R$). In particular, \eqref{expocon} does not
hold.
\end{lemma}

\begin{proof}
Let $M \in \R$ be arbitrary. Recall that $ T_y \phi_{N-k}$ by
definition is a magnetic translation  by $y$ with $|y|=3R$ of
$\phi_{N-k}$. By a change of variables for the particles with labels in
$\{k+1,\ldots,N\}$ we find that 
\begin{equation}\label{Ifirstexp}
  \big\|( J_k+ M )  \phi_k \otimes T_y \phi_{N-k} \big\| =\| I_R \|,
\end{equation}
where
\begin{equation}\label{IRdef}
I_R:= (\tilde{J}_k + M) \phi_k \otimes \phi_{N-k}
\end{equation}
 and
\begin{equation}\label{Jktilde}
\tilde{J}_k(z_1,...,z_N):=\sum_{i=1}^k \sum_{j=k+1}^N
\frac{1}{|z_i-z_j+y|} - \sum_{i=1}^k V_{\rho_{\phi_{N-k}}}(z_i+y) -
\sum_{j=k+1}^N V_{\rho_{\phi_k}}(z_j-y).
\end{equation}
By \eqref{Ifirstexp} it remains to prove that there exists $C>0$ independent of $M$
so that
\begin{equation}\label{IRest1}
 \|I_R\| \geq \frac{C}{R^3}.
\end{equation}
From the assumption of the lemma and the exponential decay of $\psi_m$
we obtain for $\phi_m$ that
\begin{equation}\label{meanzero}
 \int u \rho_{\phi_m}(u)du=O(e^{-cR}),\quad \text{for}\ m\in \{k,N-k\}.
\end{equation}
By normalization of $\psi_m$ and the definition of $\phi_m$ we may
choose $d>0$ such that
\begin{equation}\label{intphinein0}
\int_{B_d} d \underline{z}_1...d \underline{z}_N
|\phi_k(\underline{z}_1,...,\underline{z}_k)|^2
|\phi_{N-k}(\underline{z}_{k+1},...,\underline{z}_N)|^2
 \geq \frac{1}{2},
\end{equation}
for all $R \geq d+1$ where $B_d:=B(0,d)^N \subset \mathbb{R}^{3N}$. To prove \eqref{IRest1}, and thus the lemma, it clearly suffices to
show that
\begin{equation}\label{I2est}
\|I_R \chi_{B_d} \| \geq \frac{C}{R^3},\qquad C>0,
\end{equation}
where $C$ is independent of $M$. 

We are going to expand \eqref{Jktilde} in powers of $\frac{1}{|y|}$. To this end we first remark that
\begin{equation}\label{Taylor}
\frac{1}{|w-y|}=\frac{1}{|y|}+\frac{\widehat{y} \cdot
w}{|y|^2}+\frac{3(\widehat{y} \cdot w)^2-|w|^2}{2|y|^3} +
O\left( \frac{|w|^3}{|y|^4}\right),  \text{ uniformly in } |w| \leq 2R,
\end{equation}
where $\hat{y}=y/|y|$. 
Using this, \eqref{meanzero}, $\text{supp}\rho_{\phi_m} \subset B(0,R)$ and the definition of $V_{\rho}$
(see \eqref{vs}) we obtain for $|z_i|, |z_j| \leq d$
\begin{align}
V_{\rho_{\phi_k}}(z_j-y)
 &=k\left(\frac{1}{|y|}+\frac{\widehat{y}  \cdot
  z_j}{|y|^2}+\frac{3(\widehat{y} \cdot z_j)^2-|z_j|^2}{2
    |y|^3} \right)  +\frac{f_k(\widehat{y})}{|y|^3} +O(R^{-4}), \label{VEnt1}\\
V_{\rho_{\phi_{N-k}}}(z_i+y)
 &=(N-k)\left(\frac{1}{|y|}-\frac{\widehat{y} \cdot
  z_i}{|y|^2} +\frac{3(\widehat{y} \cdot z_i)^2-|z_i|^2}{2
    |y|^3} \right)  +\frac{f_{N-k}(\widehat{y})}{|y|^3} +O(R^{-4}),\label{VEnt2}
\end{align}
where
\begin{equation}\label{fmdef}
     f_m(\hat{x}):=\int \rho_{\phi_m}(u) \frac{3(\hat{x} \cdot
       u)^2-|u|^2}{2} du.
\end{equation}
Recall that $\phi_m$ depends on $R$ and hence the exponential decay of $\psi_m$
is needed for establishing the bound $O(R^{-4})$. Using \eqref{Taylor} again we obtain
\begin{multline}\label{Taylorij}
\frac{1}{|z_i-z_j+y|}=\left(\frac{1}{|y|}+\frac{\widehat{y} \cdot
(z_j-z_i)}{|y|^2}+\frac{3(\widehat{y} \cdot (z_j-z_i))^2-|z_j-z_i|^2}{2|y|^3} \right) \\+
O\left( R^{-4} \right), \quad \forall z_i, z_j \in B(0,d).
\end{multline}
Inserting \eqref{VEnt1}, \eqref{VEnt2} and \eqref{Taylorij} into \eqref{IRdef} (see also \eqref{Jktilde})
an elementary but somewhat lengthy calculation gives that
\begin{multline}\label{Ichiexp}
   I_R \chi_{B_d}(z_1,\ldots,z_N)  =\frac{1}{|y|^3} \left(\sum_{i=1}^k \sum_{j=k+1}^N
   (z_i \cdot z_j-3(z_i \cdot \hat{y})(z_j \cdot \hat{y})) + C_{y,M} \right) \\
   \times (\phi_k \otimes \phi_{N-k}) \chi_{B_d}(z_1,...,z_N) +O(R^{-4}),
\end{multline}
where $C_{y,M}=-(N-k) f_k(\hat{y})-k
f_{N-k}(\hat{y})-k(N-k)|y|^2+M|y|^3$ depends on $y$ and $M$
only. We recognize in \eqref{Ichiexp} the interaction energy $(z_i
\cdot z_j-3(z_i \cdot \hat{y})(z_j \cdot \hat{y}))/|y|^3$ of two
dipoles $z_i$ and $z_j$ separated by $y$. Let 
\begin{multline}\label{Ldef} L(\hat{y},D)=\int_{B_d} d \underline{z}_1 d \underline{z}_2... d \underline{z}_N
   \Bigg|\Big(\sum_{i=1}^k \sum_{j=k+1}^N (z_i \cdot z_j-3(z_i \cdot \hat{y})(z_j \cdot \hat{y})) + D \Big)\\
 \times \phi_k(\underline{z}_1,...,\underline{z}_k) \phi_{N-k}( \underline{z}_{k+1},...,\underline{z}_N)\Bigg|^2.
\end{multline}
Then, by \eqref{Ichiexp},
\begin{equation}\label{I2preest}
     \|I_R\chi_{B_d}\| = \frac{1}{|y|^3} L(\hat{y},C_{y,M})^{1/2} + O(R^{-4}),
\end{equation}
and it remains to show that there exists a constant $C>0$ such that
\begin{equation}\label{Lgeqc}
L(\hat{y},D) \geq C, \quad \forall \hat{y},D.
\end{equation}
This estimate together with  \eqref{I2preest}
concludes the proof of \eqref{I2est} and therefore of Lemma
\ref{noexpocon}.

To prove \eqref{Lgeqc} we first establish that $L(\hat{y},D)$ is
everywhere positive. To this end we fix $y,D$ and we consider the
function
\begin{equation*}
f(z_1,z_2,...,z_N)=\bigg(\sum_{i=1}^k z_i \bigg) \cdot \bigg(
\sum_{j=k+1}^N z_j \bigg) -3 \bigg( \sum_{i=1}^k z_i
\cdot \hat{y} \bigg) \bigg( \sum_{j=k+1}^N z_j \cdot
\hat{y} \bigg) + D,
\end{equation*}
which is part of the integrand in \eqref{Ldef}. One can show that
$f(z_1,z_2,...,z_N) \neq 0$ almost everywhere, which together with
\eqref{intphinein0}  implies that
\begin{equation}\label{Lbig0}
L(\hat{y},D) >0.
\end{equation}
Now we will use a continuity argument to show \eqref{Lgeqc}. Since
$f$ is continuous and thus bounded on $B_d$ it follows, by the
dominated convergence theorem, that $L$ is a continuous function of
$\hat{y},D$. Moreover, $\lim_{|D| \rightarrow \infty}
L(\hat{y},D)=\infty$.
 These observations together with \eqref{Lbig0} and the fact that
      a continuous function on a compact set attains its minimum
      give \eqref{Lgeqc}.
\end{proof}

\appendix
\section{Properties of $\cE^N$}

\begin{lemma}\label{minbound}\text{}
\begin{itemize}
\item[(a)] The functional $\mathcal{E}^N$ is bounded from below on the set
$S_N:=\{\Psi \in \cQ_{N,A}: \|\Psi\|=1\}$.

 \item[(b)] Every minimizing sequence $(\Psi_k)$ of $\mathcal{E}^N$ on
$S_N$ is bounded in $\cQ_{N,A}$.

\item[(c)] If for a minimizing sequence $(\Psi_k)$ of $\mathcal{E}^N$ we
have that $\Psi_k \rightarrow \Psi$ weakly in $\cQ_{N,A}$ and
strongly in $\mathcal{H}_N$ then $\Psi$ is a minimizer of
$\mathcal{E}^N$.
\end{itemize}
\end{lemma}

\begin{proof}
(a),(b) They follow from the Hardy and diamagnetic inequalities.
 We recall from \cite{LL} that the diamagnetic inequality states that if
$\phi \in H_A^1 $ then we have that $|\nabla|\phi|(x)| \leq |(D_{A,
x_1} \phi(x),...,D_{A, x_N} \phi(x))|$, for almost all $x \in
\mathbb{R}^{3N}$. In the case $N=1$ and without spin a detailed
proof of parts (a) and (b) of the Lemma is given in \cite{GHW} and
in the general case the argument is similar.

(c) We will now show that if $\Psi_k \rightarrow \Psi$ weakly in
$\cQ_{N,A}$ and strongly in $\mathcal{H}_N$ then
\begin{equation}\label{psimini}
 \mathcal{E}^{N}(\Psi) \leq \liminf_{k \rightarrow \infty} \mathcal{E}^{N}(\Psi_k),
\end{equation}
from which we conclude that $\Psi$ is a minimizer of the
Pekar-Tomasevich functional $\mathcal{E}^N$. Indeed, recall that
\begin{equation}\label{PTred}
\mathcal{E}^N(\Psi)=\langle \Psi, \left(\sum_{j=1}^N D_{A,x_j}^2+
\sum_{i<j} \frac{1}{|x_i-x_j|}\right)\Psi \rangle-D( \rho_{\Psi}).
\end{equation}
Since $\Psi_k \rightarrow \Psi$ weakly in $\cQ_{N,A}$ and
$\|\Psi_k\|=\|\Psi\|=1$ we see that
\begin{equation}\label{psimini1}
\langle \Psi, \sum_{j=1}^N D_{A,x_j}^2 \Psi \rangle \leq \liminf_{k
\rightarrow \infty} \langle \Psi_k, \sum_{j=1}^N D_{A,x_j}^2 \Psi_k
\rangle.
\end{equation}
On the other hand, since $\Psi_k \rightarrow \Psi$ weakly in
$\cQ_{N,A}$ and since $ |x_i-x_j|^{-1}$ is a bounded operator from
$\cQ_{N,A}$
to $\mathcal{H}_N$ we obtain that
\begin{equation*}
 |x_i-x_j|^{-1} \Psi_k \rightarrow  |x_i-x_j|^{-1}\Psi, \text{ weakly in }  \mathcal{H}_N.
\end{equation*}
Since, moreover, $\Psi_k \rightarrow \Psi$ strongly in
$\mathcal{H}_N$ we conclude that
\begin{equation}\label{psimini2}
  \langle \Psi,  |x_i-x_j|^{-1} \Psi \rangle =
  \lim_{k \rightarrow \infty} \langle \Psi_k,  |x_i-x_j|^{-1} \Psi_k \rangle.
\end{equation}
In addition,
\begin{equation}\label{psimini33}
D(\rho_{\Psi})-D(\rho_k)=D(\rho_{\Psi}-\rho_k,\rho_{\Psi})+D(\rho_k,\rho_{\Psi}-
\rho_{k}).
\end{equation}
We will show that
\begin{equation}\label{psimini31}
D(\rho_k,\rho_{\Psi}- \rho_{k}) \rightarrow 0.
\end{equation}
Indeed, using \eqref{ddef} and \eqref{vs} we obtain that
\begin{equation}\label{ddifdeco}
D(\rho_k,\rho_{\Psi}- \rho_{k})=\int V_{\rho_k} (\rho_{\Psi}-\rho_k)
dx.
\end{equation}
But using Lemma \ref{minbound} (b) together with \eqref{vs},
\eqref{v4} with $r=\infty$, and \eqref{Hardiamcons}
 we can prove that
\begin{equation}\label{vrkbound}
\sup_k \|V_{\rho_k}\|_{L^\infty}< \infty.
\end{equation}

 Since $\Psi_k \rightarrow
\Psi$ in $\mathcal{H}_N$ we obtain that
$\|\rho_{\Psi}-\rho_k\|_{L^1} \rightarrow 0$  which together with
\eqref{ddifdeco} and \eqref{vrkbound} implies \eqref{psimini31}.
Similarly,
\begin{equation}\label{psimini32}
D(\rho_{\Psi}-\rho_k,\rho_{\Psi}) \rightarrow 0.
\end{equation}
Combining \eqref{psimini31}, \eqref{psimini32} and
\eqref{psimini33} we obtain that
\begin{equation}\label{psimini3}
D(\rho_{\Psi})=\lim_{k \rightarrow \infty} D(\rho_k).
\end{equation}
The relations  \eqref{PTred}, \eqref{psimini1}, \eqref{psimini2} and
\eqref{psimini3} give \eqref{psimini} as desired.
\end{proof}


\end{document}